\documentclass[journal,comsoc]{IEEEtran}
\usepackage{graphicx}
\usepackage{amssymb}
\usepackage{amsmath}
\usepackage{cite}
\usepackage{mathrsfs}
\usepackage[displaymath,mathlines]{lineno}
\usepackage{color}
\usepackage{tabulary}
\usepackage{multirow}
\usepackage{pbox}
\usepackage{multicol}
\usepackage{lipsum}
\usepackage{float}
\usepackage{amsthm}
\usepackage{relsize}
\usepackage{lipsum}
\usepackage{epstopdf}
\usepackage{mathtools}
\usepackage{calc}
\usepackage{makecell}
\usepackage{caption}
\usepackage{subcaption}
\usepackage{algorithm}
\usepackage[noend]{algpseudocode}

\algrenewcommand\algorithmicforall{\textbf{foreach}}
\algrenewcommand\algorithmicindent{.8em}
\usepackage{soul}

\makeatletter
\newcommand{\doublewidetilde}[1]{{%
		\mathpalette\double@widetilde{#1}}}
\newcommand{\double@widetilde}[2]{%
		\sbox\z@{$\m@th#1\widetilde{#2}$}%
		\ht\z@=.5\ht\z@
		\widetilde{\box\z@}}
\makeatother

\newtheorem{lemma}{Lemma}

\begin{document}

\title{\huge Fast Beam Placement for Ultra-Dense LEO Networks}

\author{Trinh Van Chien, Nguyen Minh Quan,  Tri Nhu Do, Cuong Le, Tan N.  Nguyen, and Symeon Chatzinotas \vspace{-0.5cm}

\thanks{T. V. Chien and N. M. Quan are with the School of Information and Communication Technology, Hanoi University of Science and Technology, Hanoi 100000, Vietnam (e-mail: chientv@soict.hust.edu.vn and quan.nm200508@sis.hust.edu.vn).  T. N. Do is with the Department of Electrical Engineering, Polytechnique Montreal, Montreal, QC H3T 1J4, Canada (e-mail: tri-nhu.do@polymtl.ca). Tan N. Nguyen is with the Communication and Signal Processing Research Group, Faculty of Electrical and Electronics Engineering, Ton Duc Thang University, Ho Chi Minh City, Vietnam (e-mail: nguyennhattan@tdtu.edu.vn). C. Le and S. Chatzinotas are with the with the Interdisciplinary Centre for Security, Reliability and Trust (SnT), University of Luxembourg, 1855 Luxembourg, Luxembourg (e-mail: vancuong.le@uni.lu and symeon.chatzinotas@uni.lu). The work of S. Chatzinotas has received funding from the European Union's HORIZON.1.2 - Marie Skłodowska-Curie Actions (MSCA) programme under grant agreement ID 101131481. Corresponding author: Tan N. Nguyen.}
}



\maketitle

\begin{abstract} Low Earth orbit (LEO) satellites has brought about significant improvements in wireless communications, characterized by low latency and reduced transmission loss compared to geostationary orbit (GSO) satellites. Ultra-dense LEO satellites can serve many users by generating active beams effective to their locations. The beam placement problem is challenging but important for efficiently allocating resources with a large number of users. This paper formulates and solves a fast beam placement optimization problem for ultra-dense satellite systems to enhance the link budget with a minimum number of active beams (NABs). To achieve this goal and balance load among beams within polynomial time, we propose two algorithms for large user groups exploiting the modified K-means clustering and the graph theory. Numerical results illustrate the effectiveness of the proposals in terms of the statistical channel gain-to-noise ratio and computation time over state-of-the-art benchmarks.
\end{abstract}


\section{Introduction}
Non-geostationary orbit (Non-GSO) satellites such as low Earth orbit (LEO) have boosted the system performance by offering non-terrestrial connectivity solutions and supporting diverse digital technologies  \cite{9741772}. These satellites have the potential to provide low latency, enhance communication speeds, and improve energy efficiency. Accordingly, there are plans to launch tens of thousands of Non-GSO satellites to establish ultra-dense mega-constellations \cite{9473743}. This will result in an opportunity to cover the Earth, characterized by a variety of frequencies and directions. Besides, resource allocation is a critical issue in satellite communications. There are remaining challenges in enhancing the efficiency of limited resources and increasing satellite throughput due to the constraints on the limited budgets. The primary objectives involve minimizing costs, reducing computation time, and controlling congestion \cite{10182679}. An optimization problem may entail a vast number of integer variables within a discrete search space, thereby resulting in its worst-case complexity being NP-hard \cite{10048927}. In this context, we address a beam placement problem, where conical beams are defined by their half-power beam width (HPBW). Our aim is to identify suitable locations within the coverage area to position these beams, optimizing their arrangement to maximize the number of users per satellite beam and determining the optimal beam center.

Recent studies have explored various approaches to  the beam placement for LEO constellations. The authors in \cite{10366793} utilized the gradient ascent and quadratic transform for an iterative-based approach to determine the beam center and transmit power. Pachler \textit{et al.} employed the graph theory and a heuristic algorithm to place the beams \cite{pachler2021static}. However, these works have not specified the active beam number or addressed the user distribution disparity among beams. This can lead to load imbalance, reducing communication efficiency. Some proposed approaches to tackle the load balancing problem include linear optimization techniques, self-organizing feature maps, and problem linearizations. Nevertheless, these methods have not classified users into individual beams and specified the minimal number of active beams. The recently proposed methods including the deterministic annealing in \cite{10118814} and the two-stage algorithm based K-means clustering in \cite{bui2022} can ensure balanced traffic load among employed beams. However, they only considered a small number of users with the  time consumption being inadequate for practical applications.

This work suggests strategies to minimize the number of active beams and improve the load balance in the network serving many users. We address the optimization process of reducing the distance between users and their corresponding serving beams. This particular topic has not been thoroughly investigated in this literature  because of the non-convex and NP-hard properties.  We concentrate on developing two algorithms to cluster users based on the Half Power Beam Width (HPBW) to obtain a solution in a polynomial time with the following key contributions: $i)$ we develop a system where multi-beams serve multi-users as well as formulating an optimization problem; $ii)$ we propose two algorithms: one utilizing loops in conjunction with K-means clustering, and another employing a two-phase algorithm based on graph theory; and $iii)$ we provide experimental results to compare the effectiveness of the two proposed algorithms and compare with the algorithms provided in the document \cite{bui2022}. Specifically, the proposed algorithms can be applied for large-scale networks with many users.

\section{Ultra-Dense Satellite Communication \& Beam Placement Optimization Problem}
This section presents a high throughput multi-beam system. A beam placement problem is then formulated and solved.
\subsection{System Model}
\begin{figure}[t]
    \centering
    \includegraphics[trim=0.5cm 0.5cm 0.5cm 0.2cm, clip=true, scale=0.15] {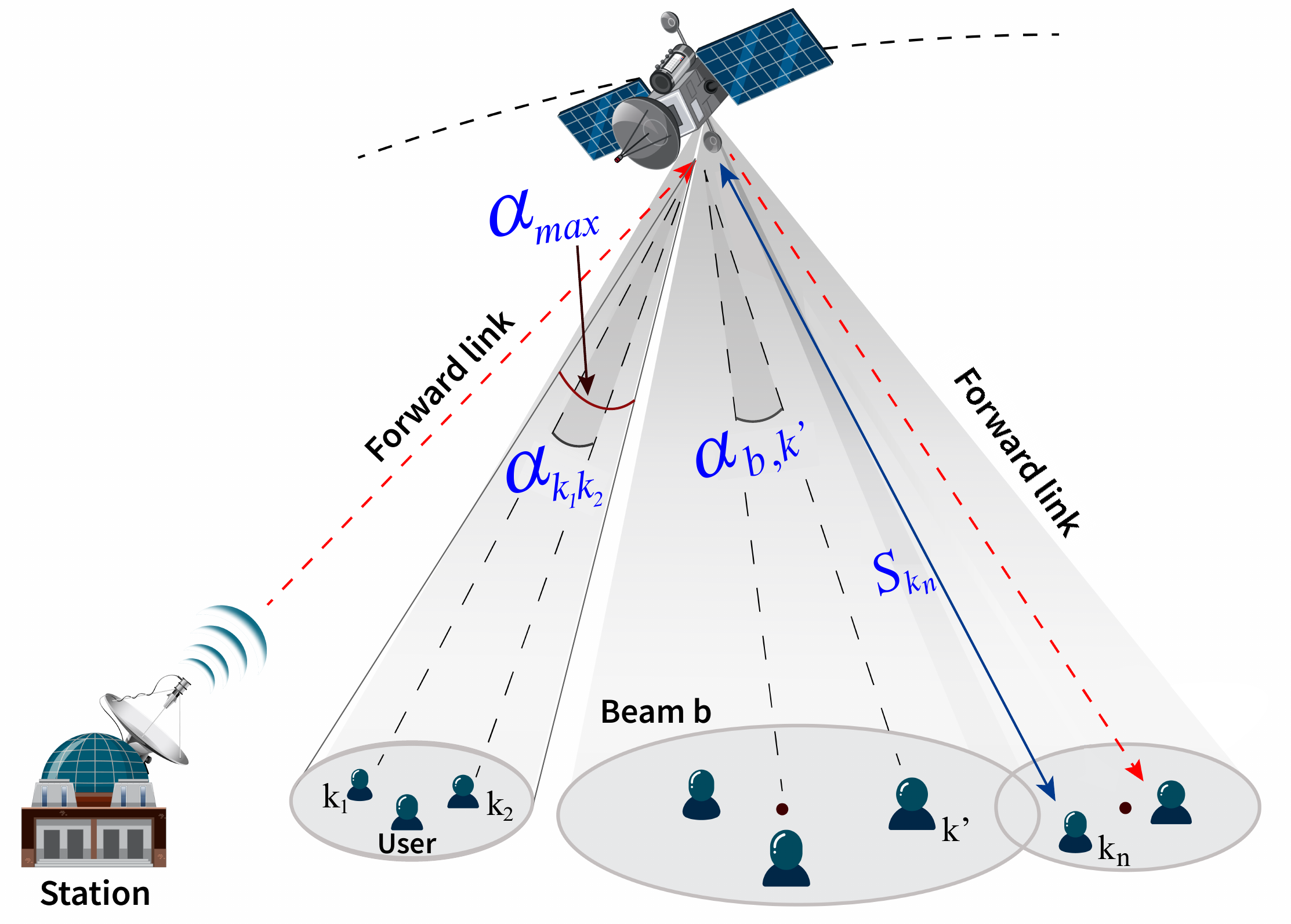}
    \caption{Illustration of the system model, where multiple users can be served by the same beams.}
    \label{fig:system}
\end{figure}
A satellite is equipped with $N$ beams having the circle pattern to serve $K$ users, where multiple users can be served by one beam. Due to energy limitations and mutual interference, it might not be efficient to activate all the beams. Let $B$ and $\mathcal{B} = \{1, 2, \ldots B\}$ be the number and the set of active beams, $\mathcal{U} = \{k_1, k_2, \ldots, k_K\}$ be the set of all users, and $\mathcal{U}_b \subseteq \mathcal{U}$ be a subset of users served by the $b$-th active beam. Then the beam placement problem is defined as assigning the set of users $\mathcal{U}$ to a subset $\mathcal{B}$ such that the NABs are minimized, where every user must be served by exactly one beam. The practical constraints are expressed as
\begin{align}
\bigcup\nolimits_{b\in\mathcal{B}}\mathcal{U}_b = \mathcal{U}; \mathcal{U}_b \neq \emptyset, \mathcal{U}_b \cap \mathcal{U}_{b'} = \emptyset, \forall b, b' \in \mathcal{B}, b \neq b'.
\end{align}
We assume that a user is served by a given beam if and only if this user is located in at least the HPBW of the beam. Let $\alpha_{b,k} \in [0, \pi/2]$ be the angle between user $k$  and the centre of beam $b$ (see Fig.~\ref{fig:system}), and $g_{bk}$ be the radio pattern between beam $b$ and user $k$. Also, $g_{\text{max}}$ is the maximum gain obtained when user $k$ is located at the beam centre. Then, user $k$ can be served by beam $b$ if and only if $g_{bk} \geq g_{\text{max}}/2$. According to the 3GPP report \cite{3gpp2019study},
$g_{bk} = g_{\text{max}}G(\alpha_{b,k}), \forall b \in \mathcal{B}, k \in \mathcal{U},$
where $G(\alpha_{b,k})$ is the normalized pattern gain, which is
\begin{equation}
    G(\alpha_{b,k}) = 
    \begin{cases}
        1, & \text{if $\alpha_{b,k} = 0$}, \\
        4 \bigg|\frac{J_1\left(\frac{2\pi}{\lambda}r\sin{(\alpha_{b,k})}\right)}{\frac{2\pi}{\lambda}r\sin{(\alpha_{b,k}})}\bigg|^2, & \text{if $0 < \alpha_{b,k} \leq \frac{\pi}{2}$},
    \end{cases}
\end{equation}
where $J_1(\cdot)$ is the Bessel function of the order one, $\lambda$ is the carrier wavelength, and $r$ is the radius of the antenna aperture. 
Let $\epsilon_{bk}$ be the antenna efficiency and $S_{k}$ be the slant range distance between user $k$ and the satellite. $D$ is the satellite antenna diameter. The statistical channel gain (SCG) at an arbitrary user is 
$G_{bk} = g_{bk}\bar{g}_{bk}/(L_{bk}L_{atm}),$
where $\bar{g}_{bk} = \epsilon_{bk}\pi^2D^2/\lambda^2$ is the received antenna power gain, $L_{bk} = 16\pi^2S_{k}^2/\lambda^2$ is the free path loss between user $k$ and the satellite, and $L_{atm}$ is the atmospheric absorption. The  statistical channel gain-to-noise ratio (SCGNR) at an user $k$ is denoted as
  $\mathrm{SCGNR}_k = G_{bk} / \sigma^2,$
where $\sigma^2$ is the noise power. 


\subsection{Problem Formulation}
Let's define decision variables $x_b = \{0, 1, \forall b \in [1, N]\}$ where $x_b=1$ as  beam $b$ is active. Otherwise, this beam is inactive. Binary variables $y_{bk} = \{0, 1, \forall b\in [1, N], k\in\mathcal{U}\}$ are introduced to assign users to the beams. Specifically, $y_{bk} = 1$ if user $k$ is assigned to beam $b$, and $y_{bk} = 0$ otherwise. The beam placement problem now becomes finding $\mathbf{x} = \{x_b, \forall b \in [1, N]\}$ and $\mathbf{y} = \{y_{bk}, \forall b \in [1, N], k \in \mathcal{U}\}$ such that the NABs is minimized. The problem is formulated as
\begin{subequations}
\begin{align}
    \underset{\mathbf{x}, \mathbf{y}}{\text{minimize}}
     \quad& \sum\nolimits_{b=1}^Nx_b \label{prob:ILP}\tag{\sf P}\\
    \text{subject to} \quad & \sum\nolimits_{b=1}^Ny_{bk} = 1, \forall k\in \mathcal{U}, \label{ctr:sum_y}\\
    & x_b \leq \sum\nolimits_{k=1}^Ky_{bk}, \forall b \in [1, N], \label{ctr:active_beam} \\
    & M(\beta_{bk} - 1) \leq g_{bk} - g_{\text{max}} / 2 \leq M\beta_{bk}, \nonumber \\
    & \quad\quad\quad\quad\quad\quad\quad\quad \forall b \in [1, N], k\in \mathcal{U},\label{ctr:beta1} \\
    & y_{bk} \leq \beta_{bk}, \forall b \in [1, N], k\in \mathcal{U},\label{ctr:half_power}\\
     & x_b, {y}_{bk}, {\beta}_{bk} \in \{0, 1\}, \forall b \in [1, N], k\in \mathcal{U},\label{ctr:binary}
\end{align}
\end{subequations}
where $M$ is a large constant and $\beta_{bk}$ are auxiliary variables to represent whether the HPBW condition is satisfied, \textit{i.e.,} $\beta_{bk} = 1$ implies that $g_{bk} \geq g_{\text{max}} / 2$ and user $k$ can be served by the $b^{th}$ beam, and $\beta_{bk} = 0$ otherwise. The objective of problem \ref{prob:ILP} is to minimize the NABs. Constraint \eqref{ctr:sum_y} is to ensure that all the users are served and each of those is served by exactly one beam. Constraint \eqref{ctr:active_beam} guarantees that no beam is activated without serving any users. Constraint \eqref{ctr:beta1} is to linearize the HPBW condition. Constraint \eqref{ctr:half_power} ensures that all users are within the coverage regions of the assigned beams.

\subsection{Connection (\ref{prob:ILP}) with the Minimum Clique Cover Problem}
\label{subsec:min_clique}
In this section, we connect our optimization problem to the well-known graph theory problem, Minimum Clique Cover, which involves partitioning the vertex set into the smallest number of disjoint cliques. Let's define an undirected graph $\mathbf{G} = (V, E)$, where $V = \mathcal{U}$ is a set of vertices and $E=\{(k_i, k_u), k_i, k_u \in V\}$ is a set of edges. Each vertex $k \in V$ corresponds to user $k$ in $\mathcal{U}$, and each edge $e = (k_i, k_u)$ indicates that two users $k_i$ and $k_u$ can be served by a single beam. We consider the HPBW condition to determine whether two users can be served by the same beam. Let $\alpha_{k_ik_u}$ be the angle from the satellite to user $k_i$ and $k_u$ ($k_i,k_u \in \mathcal{U}$), calculated as $\alpha_{k_ik_u} = \arccos({2S_{k_i}S_{k_u}}/({S_{k_i}^2 +  S_{k_u}^2 -d_{k_i,k_u}^2})).$ Where, $d_{k_i, k_u}$ is the geography distance, computed through the Spherical Law of Cosines. We note that $\alpha_{\max}$ is the HPBW of beam $b$, $\forall b \in \mathcal{B}$, i.e., the normalized pattern gain $G(\alpha_{\max}/2) = 1/2$. Then two users $k_i$ and $k_u$ can be served by the same beam $b$ if and only if $\alpha_{k_ik_u} \leq \alpha_{\max}$ (see Fig.~\ref{fig:system}). However, \cite{pachler2021static} showed that some users cannot be covered if using an HPBW of $\alpha_{\max}$. Thus, we use $\alpha_{\max}/2$ to ensure that all users in the same beam are within the HPBW area.

Then the graph $\mathbf{G}$ can be represented by an adjacency matrix $\mathbf{A}$ defined as 
$\mathbf{a}_{k_ik_u} =  1, \text{ if } \alpha_{k_ik_u} \leq {\alpha_{\max}/2}$, and $\mathbf{a}_{k_ik_u} = 0,\text{ otherwise.}$ Therefore, to ensure that all users belong to a single beam \( b_t \), the following condition must be satisfied:
$\mathbf{a}_{k_ik_u} =  1, \forall k_i,k_u \in b_t$. This means that each pair of vertices corresponding to users who belong to the same beam is connected. Thus, a subset of $V$ corresponding to the subset of users that satisfies Problem~\eqref{prob:ILP} constitutes a clique, i.e., a complete sub-graph of $\mathbf{G}$. Our problem is now equivalent to partitioning the vertex set $V$ into $B$ disjoint subsets $V_1, V_2, \ldots V_B$, where $B$ is minimized and the sub-graph induced by $V_b$ is a complete graph, $\forall 1\leq b \leq B$. This problem is known as the Minimum Clique Cover (MCC) problem, which is NP-hard \cite{karp2010reducibility}. For a network of $K$ users with a non-limited NABs, it is also NP-hard to approximate Problem \eqref{prob:ILP} within an approximation ratio of $K^{1 - \epsilon}$ for any given $\epsilon > 0$.

\section{Beam Placement with Bisection K-means}

This section proposes an algorithm that combines K-means clustering and binary search, namely Bisection K-means (BK-Means) as shown in Algorithm~\ref{pseudo:bkm} to tackle Problem~\eqref{prob:ILP}. The primary objective of the clustering algorithm is to divide a given dataset into $M$ groups, where the value of $M$ is predetermined. Additionally, we use Bisection algorithm to explore different scenarios of the necessary NABs. Given that $(\phi_{k_i}$,$\theta_{k_i})$ is longitude and latitude of user $k_i$ and $R$ is the radius of the Earth. We utilize Cartesian coordinate format denoted by $(x_{k_i},y_{k_i},z_{k_i})$ as:
$x_{k_i} = R \cos(\phi_{k_i}) \cos(\theta_{k_i}), y_{k_i} = R \cos(\phi_{k_i}) \sin(\theta_{k_i}),  z_{k_i} = R \sin(\phi_{k_i}), \forall i \in [1,K]$ to recalculate the distance between $k_i$ and $k_u$ ($\forall k_i,k_u \in \mathcal{U}$) which is employed in K-means clustering as $d_{k_i,k_u}=\sqrt{(x_{k_i}-x_{k_u})^2 + (y_{k_i}-y_{k_u})^2 + (z_{k_i}-z_{k_u})^2 }$. The algorithm involves three main steps: $i)$ determining the number of clusters 
$B$ for testing; $ii)$ performing K-means clustering with $B$ clusters repeatedly until a solution is found or a maximum of $\mu$ iterations is reached; and $iii)$ repeating steps $i)$ and $ii)$ until the termination condition is met. From set of $N$ beams, we divide range of beams into 2 ranges with a mean value $B_{\text{mean}}$, this value is used to test the feasibility of NABs, and then making a decision to resize the range accordingly. The NABs can be in $[1; B_{\text{mean}}]$ only if a solution with $B_{\text{mean}}$ beams is feasible. Otherwise, it falls within the range of $[B_{\text{mean}}+1; N]$. To evaluate the feasibility of $B_{\text{mean}}$ beams, we utilize the adjacency matrix $\mathbf{A}$ of graph $\mathbf{G}$. We denote $\mathcal{S} = \{\mathcal{U}_1, \mathcal{U}_2, \ldots, \mathcal{U}_{B_{\text{mean}}}\}$, $\mathcal{U}_i \in \mathcal{S}$, $i \in \{1,2,\ldots, B_{\text{mean}}\}$ as a solution of problem~\eqref{prob:ILP}, so $\mathrm{feasible}$ can be concluded as
\begin{equation}
\mathrm{feasible} =  
\begin{cases}
    \text{False}, & \text{if } \exists (k_i,k_u) \in \mathcal{U}_i\text{ s.t } \mathbf{a}_{k_ik_u} = 0\\
    \text{True}, & \text{otherwise.}
\end{cases}
\end{equation} 

\textit{Computational complexity:} Time complexity of BK-Means is estimated as $\mathcal{O}(\mu KN\log(N)(NI+K))$, with $N$ maximum number of beams, $K$ users, $I$ iterations of K-means clustering, and $\mu$ maximum number of iterations. 
\begin{algorithm}
\caption{Bisection K-Means clustering based beam placement (BK-Means)}
\label{pseudo:bkm}
\begin{algorithmic}[1]

\State \textbf{Input}: Dataset, graph $\mathbf{G}$, matrix $\mathbf{A}$, number of maximum beams $N$, maximum iterations $\mu$.
\State \textbf{Output}: List of active beams $\mathcal{B}$ and user assignment $\mathcal{S}$.
\State Initialize $B_{\min} = 1$, $B_{\max} = N$;
        \While {$B_{\min} + 1 < B_{\max}$}
            \State Set $B_{\text{mean}} = \lfloor(B_{\min} + B_{\max})/2\rfloor$;
            \State $\mathrm{feasible} = \mathrm{False}$; $\mathrm{count} = 0$;
            \While{not $\mathrm{feasible}$ and $\mathrm{count} < \mu$}
                \State Run K-means with $B_{\text{mean}}$ clusters;
                \If{every cluster forms a clique of $\mathbf{G}$}
                    \State $\mathrm{feasible} = \text{True}$;
                \Else
                    \State $\mathrm{count} = \mathrm{count} + 1$;
                \EndIf
            \EndWhile
            \If{$\mathrm{feasible}$}
                \State $B_{\max} = B_{\text{mean}}$;
            \Else
                \State $B_{\min} = B_{\text{mean}}$;
            \EndIf
        \EndWhile
\State \textbf{Return}: $B$ subsets of users corresponding to $B$ beams;
\end{algorithmic}
\end{algorithm}
\section{Graph Coloring Based Beam Placement}
This section proposes a two-phase algorithm: first, it solves MCC problem to get an initial solution, then iteratively refines this solution for better load balancing among active beams.

\textbf{\textit{Phase 1 (solution initialization)}}: Our approach is transforming the MCC problem into the equivalent Graph Coloring (GC) problem \cite{karp2010reducibility}. The GC problem involves assigning colors to the vertices of the graph in such a way that no two adjacency nodes have the same color, and the number of colors is minimized. Let $\overline{\mathbf{G}} = (V, \overline{E})$ be the complement graph of $\mathbf{G}$. That is, $\overline{\mathbf{G}}$ is constructed on the same vertex set of $\mathbf{G}$ and two vertices of $\overline{\mathbf{G}}$ is connected if and only if they are independent in $\mathbf{G}$. Let $\overline{\mathbf{A}}$ be the adjacency matrix of $\overline{\mathbf{G}}$ given as:   
\begin{equation}
\overline{{a}}_{k_ik_u} =  
\begin{cases}
    1, & \text{if }a_{k_ik_u} = 0, \\
    0, & \text{otherwise.}
\end{cases}
\label{eq:complement_A}
\end{equation}
Finding minimum clique cover of $\mathbf{G}$ is equivalent to finding a proper vertex coloring for $\overline{\mathbf{G}}$ such that the number of colors is minimized. Each subset of vertices in $\overline{\mathbf{G}}$ that shares the same color as an active beam in our problem. To solve this problem, we use a greedy algorithm which starts with a list of vertices of $\overline{\mathbf{G}}$ sorted in descending order of their degrees. The algorithm iteratively assigns the smallest color to a vertex and only uses a new color if all previously used colors are already assigned to its neighbors, ensuring that adjacency vertices have different colors. 

\textbf{\textit{Phase 2 (load balancing improvement)}}: The initial solution is iteratively refined by moving users from high-workload beams to lower-workload beams. The algorithm considers all the beam pairs, checks and re-assigns unbalanced assignments. The process is repeated until no further feasible reassignments are possible. More details of these steps are described in \textit{Phase 2} of Algorithm \ref{pseudo:tgbp} with the convergence as follows. 

\begin{lemma}
    \label{lemma:phase2_convergence}
    \textit{Phase 2} of Algorithm~\ref{pseudo:tgbp} converges after $BK$ movements, with the NABs ($B$) and the number of users ($K$).
\end{lemma}
\begin{proof}
    Consider $\mathcal{S} = \{\mathcal{U}_1, \mathcal{U}_2, \ldots, \mathcal{U}_B\}$ being an result of \textit{Phase 1}. Let $u_b = |\mathcal{U}_b|, \forall \mathcal{U}_b \in \mathcal{S}$. Assuming that the set $S$ is sorted in the descending order of its elements' sizes, i.e., $u_1 \geq u_2 \geq \ldots \geq u_B$. Let define $\Phi_b = \sum_{j = b+1}^{B}(u_b - u_j)$ be the balancing indicator of beam $b$ and $\Phi = \sum_{b=1}^{B}\Phi_b$ be the balancing indicator of the whole system. We have
    \begin{equation}
    \begin{split}
        \Phi &= \sum\nolimits_{b=1}^{B}\sum\nolimits_{j = b+1}^{B}(u_b - u_j) = \sum\nolimits_{b=1}^{B}\left((B - b)u_b \right. \\ & \left. - \sum\nolimits_{j = b+1}^{B}u_j\right) \leq B\sum\nolimits_{b=1}^{B}u_b = B\sum\nolimits_{b=1}^{B}|\mathcal{U}_b| = BK.
    \end{split}
    \end{equation}
    Since we only move a user from group $\mathcal{U}_b$ to group $\mathcal{U}_{b'}$ if $|\mathcal{U}_b| - |\mathcal{U}_{b'}| > 1$, the balancing indicator $\Phi$ decreases by at least one after each movement. Consequently, Algorithm \ref{pseudo:tgbp} finishes within a maximum of $BK$ steps.
\end{proof}
\begin{algorithm}
    
    \caption{Two-phase Greedy Beam Placement (TGBP)}
    \label{pseudo:tgbp}
    \begin{algorithmic}[1]
        \State \textbf{Input}: List of users $\mathcal{U}$, graph $\overline{\mathbf{G}}$ and matrix $\overline{\textbf{A}}$.
        \State \textbf{Output}: List of active beams $\mathcal{B}$ and user assignment $\mathcal{S}$.
        \% \textit{Phase 1: Solution initialization}
        \State Sort all users $\mathcal{U}$ in descending order of degrees in $\overline{\mathbf{G}}$.
        \State Initialize $\mathcal{B} = \emptyset$, index $b = 0$, user assignment $\mathcal{S} = \emptyset$;
            \ForAll{user $k$ in the sorted list}
                \If{$k$ is not covered by any beam}
                    \State $\mathcal{U}_b = \{k\}$; Mark user $k$ as covered;
                    \ForAll{user $u$ after $k$ in the sorted list}
                        \If{$u$ is not covered and $\overline{\mathbf{a}}_{uu'} = 0~\forall u' \in \mathcal{U}_b$}
                            \State $\mathcal{U}_b = \mathcal{U}_b \bigcup \{u\}$; Mark user $u$ as covered;
                        \EndIf
                    \EndFor
                    \State $\mathcal{S} = \mathcal{S} \bigcup \{\mathcal{U}_b$\}; $b = b + 1$;
                \EndIf
            \EndFor
        \State \textbf{Return}: $\mathcal{B}$, $\mathcal{S}$;\\
        \% \textit{Phase 2: Load balancing improvement}
        \State $\mathrm{stop}$ = False;
        \While{not $\mathrm{stop}$}
            \State$\mathrm{count}\_\mathrm{moves} = 0$;
            \ForAll{pair of user groups $\mathcal{U}_b, \mathcal{U}_{b'} \in \mathcal{S}$}
                \If{the load balancing gap $|\mathcal{U}_b| - |\mathcal{U}_{b'}| > 1$}
                    \ForAll{user $k$ in $\mathcal{U}_b$}
                        \If{$\mathbf{a}_{kk'} = 1~\forall k' \in \mathcal{U}_{b'}$}
                            \State Move user $k$ from $\mathcal{U}_b$ to $\mathcal{U}_{b'}$;
                            \State $\mathrm{count}\_\mathrm{moves} = \mathrm{count}\_\mathrm{moves} + 1$;

                            \If{$|\mathcal{U}_b| - |\mathcal{U}_{b'}| \leq 1$}
                                \State break;
                            \EndIf
                        \EndIf
                    \EndFor
                \EndIf
            \EndFor
            \If{$\mathrm{count}\_\mathrm{moves} = 0$}
                \State $\mathrm{stop}$ = True;
            \EndIf
            \EndWhile
    \State \textbf{Return}: $\mathcal{B}$, $\mathcal{S}$;
    \end{algorithmic}
\end{algorithm}
\textit{Computational complexity:} The complexity of Algorithm~\ref{pseudo:tgbp} is $\mathcal{O}(N^3K^2)$ with $K$ users and $N$ maximum number of beams. Overall, we provide the performance guarantee for the solutions delivered by Algorithm \ref{pseudo:tgbp} as in Lemma~\ref{lemmav1}. 
\begin{lemma} \label{lemmav1}
    The NABs ($B$) returned by Algorithm \ref{pseudo:tgbp} is bounded by $B \leq K + 1 - \delta(\mathbf{G})$, where $\delta(\mathbf{G}) = \min\left\{\deg_{\mathbf{G}}(v) : v\in V\right\}$ is the minimum degree over all vertices of the graph $\mathbf{G}$.
\end{lemma}
\begin{proof}
    Since the number of user groups keeps unchanged in \textit{Phase 2}, the NABs returned by Algorithm~\ref{pseudo:tgbp} only depends on \textit{Phase 1}. Let $\Delta(\overline{\mathbf{G}}) = \max\left\{\deg_{\overline{\mathbf{G}}}(v) : v\in V\right\}$ be the maximum degree over all vertices. Following \cite{arumugam2016handbook}, the list of active beams returned by \textit{Phase 1} satisfying that $|\mathcal{B}|~\leq~\Delta(\overline{\mathbf{G}})~+~1$. Moreover, since $\overline{\mathbf{G}}$ is the complement graph of $\mathbf{G}$, we have $\deg_{\overline{\mathbf{G}}}(v) = |V| - \deg_{\mathbf{G}}(v), \forall v\in V$. In summary, we have: 
    \begin{align}
        |\mathcal{B}| &\leq \max\left\{\deg_{\overline{\mathbf{G}}}(v) : v\in V\right\} + 1 \nonumber\\
        &= \max\left\{|V| - \deg_{\mathbf{G}}(v) : v\in V\right\} + 1 \nonumber\\
        &=|V| + 1 + \max\left\{- \deg_{\mathbf{G}}(v) : v\in V\right\} \nonumber\\
        &=K + 1 - \min\left\{\deg_{\mathbf{G}}(v) : v\in V\right\}, 
    \end{align}
    which completes the proof.
\end{proof}

\section{Numerical Results}
We investigate a LEO satellite system providing simultaneous service to multiple users, which are located in a longitude range of $[30^\circ, 40^\circ]$ and a latitude range of $[-120^\circ, -110^\circ]$. The satellite is positioned at coordinates with the latitude and longitude of $[0^\circ, -88.7^\circ]$, and the satellite antenna  has a diameter of 0.6 [m]. Each beam center of the satellite has a maximum gain of 50 [dBi] and a HPBW with a maximum angle of $\alpha_{\text{max}} = 3.2^\circ$. The carrier frequency is 18.05 [GHz], and the aperture radius is $5\lambda$. The maximum iteration is $\mu = 200$, and the maximum iteration number of K-means in Algorithm~\ref{pseudo:bkm} (line 6) is $I = 500$. All the benchmarks  are implemented in Matlab with the same configurations: Windows 11 Pro operating system, Intel Core i$7-12700$ processor, $2.10$GHz CPU, and $64$GB RAM.
We conducted experiments to compare the performance of four algorithms: a heuristic method based on graph theory in \cite{pachler2021static}, referred to as HCBP; K-means with 2 phases (TPK-Means) in \cite{bui2022}, BK-Means, and TGBP. We note that TPK-Means and HCBP from previous studies have a significant limitation due to their high computational time, making them unsuitable for large user groups. Therefore, to compare the performance of our two algorithms with TPK-Means and HCBP, we consider two cases: small user sets (10-30 users) and large user sets (50-1000 users), with TPK-Means and HCBP only applicable to small user sets.

Fig.~\ref{fig:fig2}(a) and Fig.~\ref{fig:fig2}(b) show the average NABs required for power distribution using HCBP, TPK-Means, BK-Means, and TGBP. In small user sets (Fig.~\ref{fig:fig2}(a)), the NABs are different across algorithms, TGBP is more efficient overall. For larger user sets (Fig.~\ref{fig:fig2}(b)), TGBP outperforms BK-Means as the user count increases. This is due to TGBP’s Greedy Algorithm, which finds the minimum NABs sequentially, unlike the iterative K-Means-based approaches of BK-Means and TPK-Means. HCBP's incremental algorithm, which selects random cliques by size, also struggles to guarantee optimal solutions.

\begin{figure*}[t]
	\centering
    \begin{minipage}[t]{0.245\textwidth}
        \includegraphics[trim=0.25cm 0.1cm 0.5cm 0.7cm, clip=true, scale=0.33]{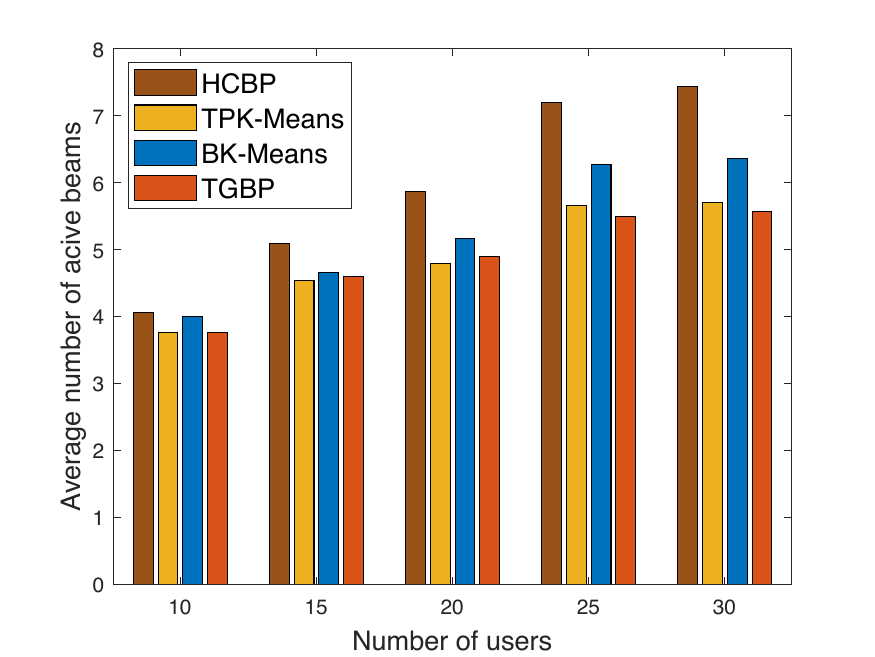}  
        \subcaption{}  	
        \label{fig:number_of_beams_small}
    \end{minipage}
    \begin{minipage}[t]{0.245\textwidth}
	\includegraphics[trim=0.25cm 0.12cm 0.5cm 0.7cm, clip=true, scale=0.33]{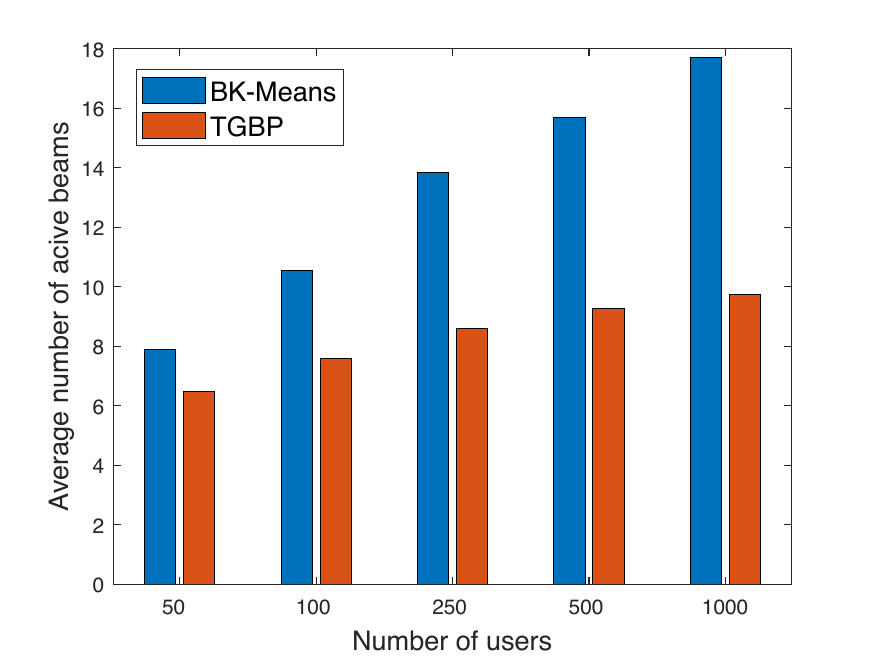}
        \subcaption{}
        \label{fig:number_of_beams_large}
    \end{minipage}
    \begin{minipage}[t]{0.245\textwidth}
	\includegraphics[trim=0.25cm 0.1cm 0.5cm 0.7cm, clip=true, scale=0.33]{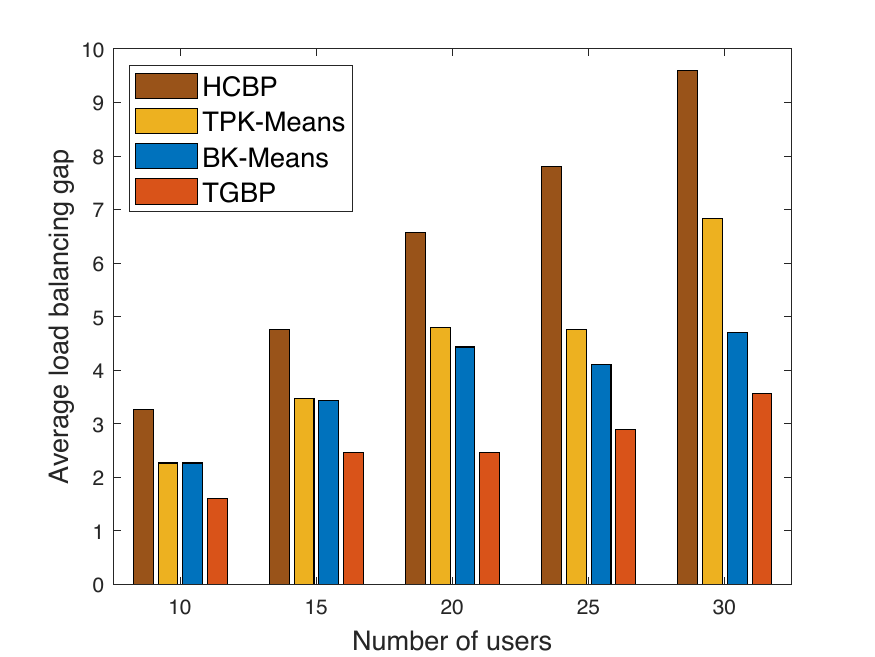}
        \subcaption{}
        \label{fig:balancing_gap_small}
    \end{minipage}
     \begin{minipage}[t]{0.245\textwidth}
	\includegraphics[trim=0.25cm 0.1cm 0.5cm 0.7cm, clip=true, scale=0.33]{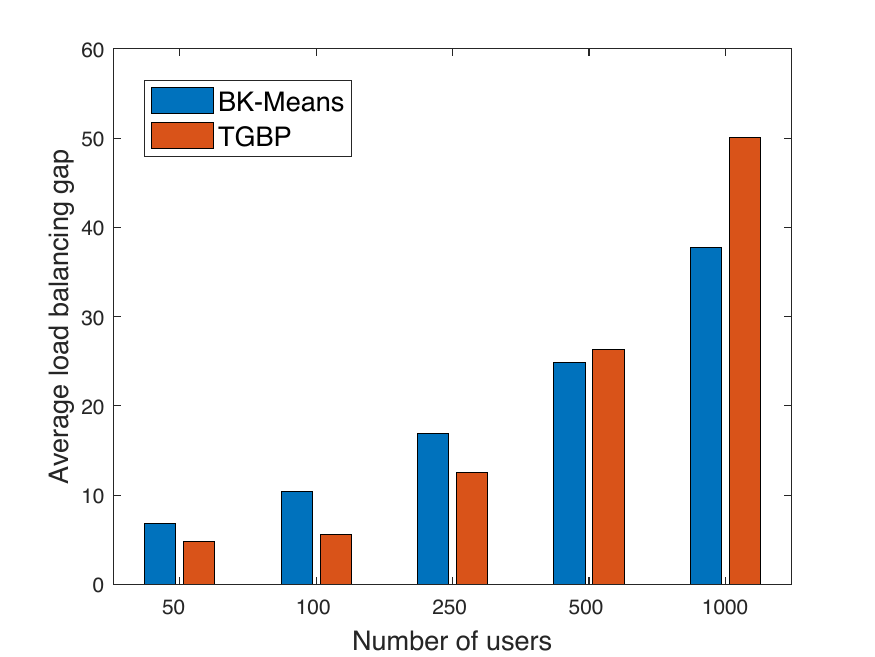}
        \subcaption{}
        \label{fig:balancing_gap_large}
    \end{minipage}
\vspace{-0.2cm}
\caption{(a) Number of active beams on small-scale networks. (b) Number of active beams on large-scale networks. (c) Average load balancing gap on small-scale networks. (d) Average load balancing gap on large-scale networks. }
\label{fig:fig2}
\vspace{-0.3cm}
\end{figure*}

\begin{figure*}[t]
	\centering
    \begin{minipage}[t]{0.245\textwidth}
        \includegraphics[trim=0.25cm 0.1cm 0.5cm 0.7cm, clip=true, scale=0.315]{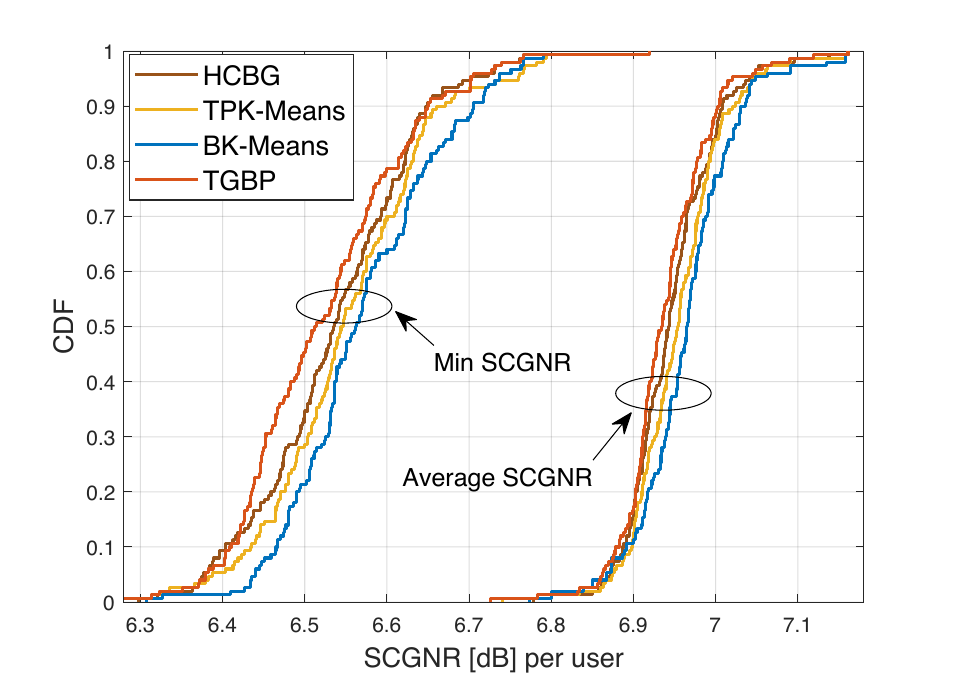}    	
        \subcaption{}
        \label{fig:SNR_small}
    \end{minipage}
    \begin{minipage}[t]{0.245\textwidth}
	\includegraphics[trim=0.25cm 0.0cm 0.5cm 0.7cm, clip=true, scale=0.315]{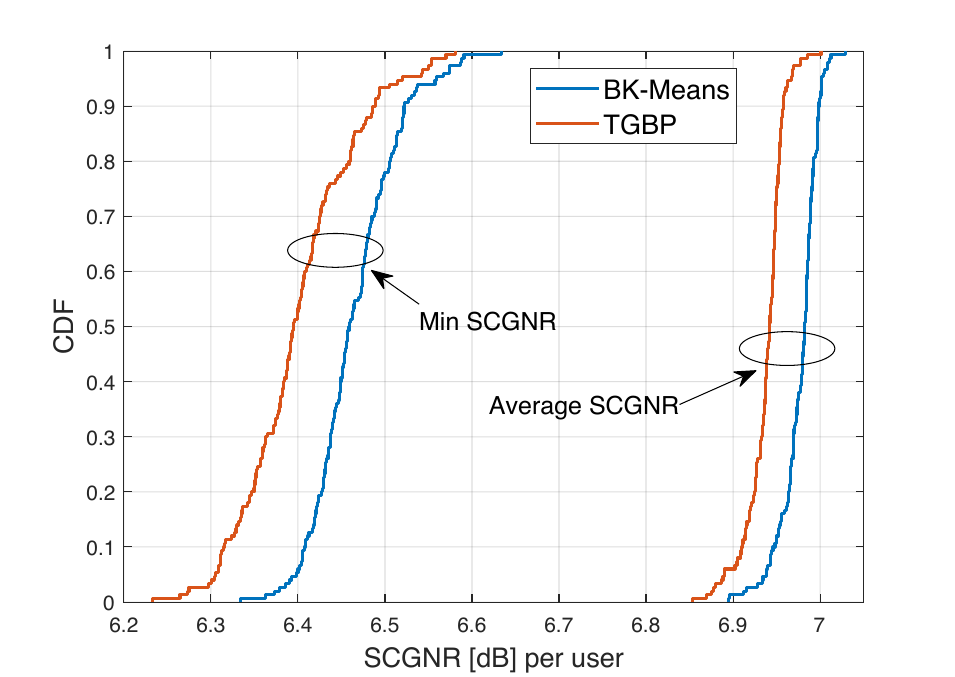}
        \subcaption{}
        \label{fig:SNR_large}
    \end{minipage}
    \begin{minipage}[t]{0.245\textwidth}
	\includegraphics[trim=0.25cm 0.0cm 0.5cm 0.7cm, clip=true, scale=0.325]{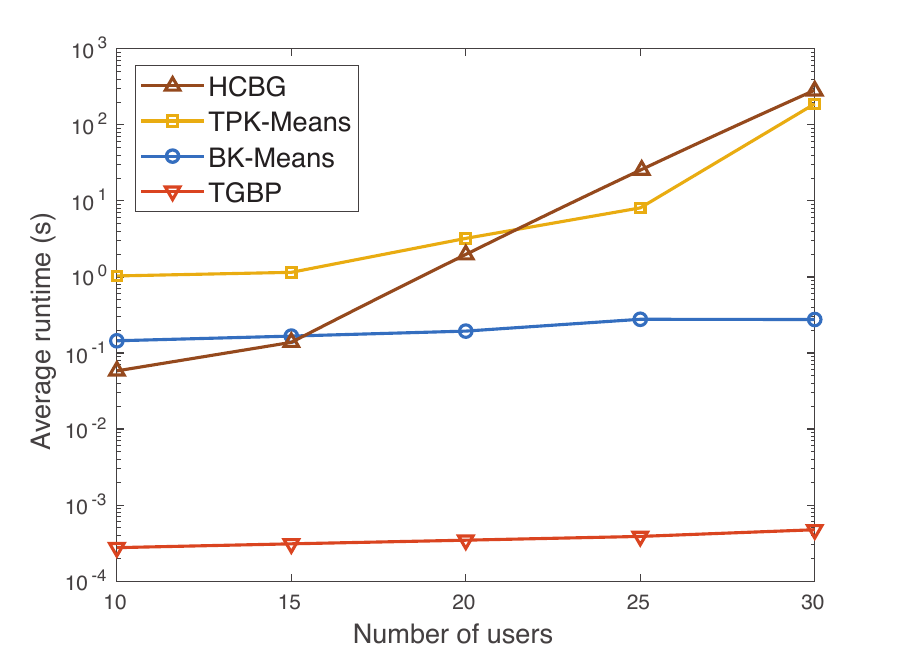}
        \subcaption{}
       \label{fig:computing_time_small}
    \end{minipage}
     \begin{minipage}[t]{0.245\textwidth}
	\includegraphics[trim=0.25cm 0.1cm 0.5cm 0.7cm, clip=true, scale=0.325]{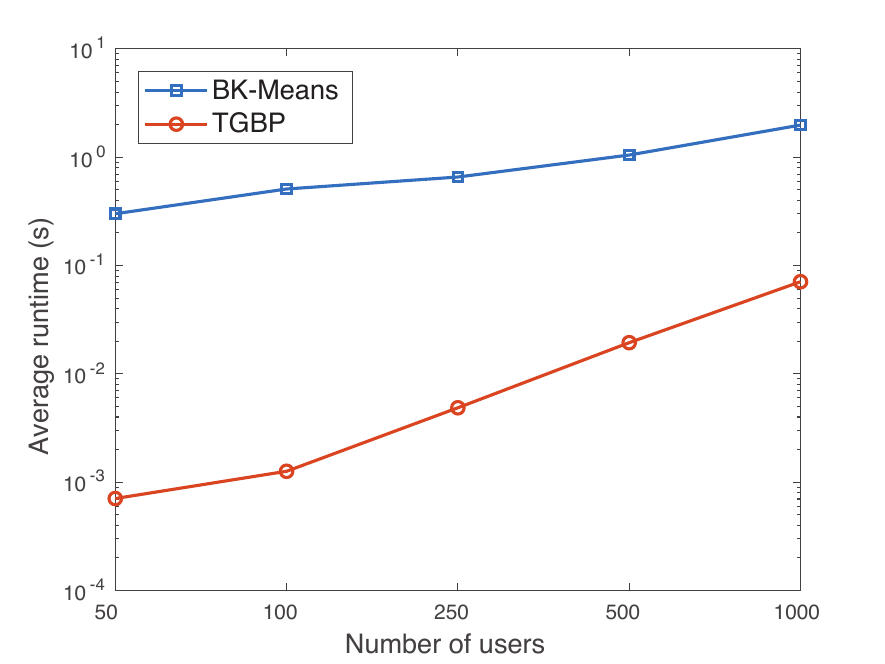}
        \subcaption{}
       \label{fig:computing_time_large}
    \end{minipage}
\vspace{-0.2cm}
\caption{(a) CDF of the average SCGNR and min SCGNR [dB] per user on small-scale networks. (b) CDF of the average SCGNR and min SCGNR [dB] per user on large-scale networks. (c) Computing time required on small-scale networks. (d) Computing time required on large-scale networks. }
\label{fig:fig3}
\vspace{-0.5cm}
\end{figure*}

Fig.~\ref{fig:fig2}(c) and Fig.~\ref{fig:fig2}(d) depict the average load balancing gap, which is the average difference between the largest beam (serving the most users) and the smallest beam (serving the least users) resulting from all algorithms across various user sets. In Fig.~\ref{fig:fig2}(c), there is a noticeable difference among the four algorithms as the user count increases. Specifically, TGBP proves to offer a more balanced distribution compared to BK-Means, TPK-Means and HCBP. This performance is due to the refining phase of TGBP, which enhances balance in smaller user sets through multiple refining steps. However, with larger user sets (500 and 1000 users), BK-Means demonstrates a better average load balancing gap compared to TGBP, as illustrated in Fig.~\ref{fig:fig2}(d). This shows that the balancing efficiency of TGBP decreases as the number of refining steps increases with larger user sets. Conversely, BK-Means benefits from its iterative clustering approach, which better manages larger datasets and maintains a more consistent load balancing gap.

Fig.~\ref{fig:fig3}(a) and Fig.~\ref{fig:fig3}(b) show the cumulative distribution function (CDF) of the statistical channel gain to noise ratio (SCGNR) [dB] for each user. $i.e., 10{\text{ log}}_{10}(G_{b,k}/\sigma_{k}^2), \forall b \in \mathcal{B}, k \in \mathcal{U} $. BK-Means provides benefits in both average and minimum SCGNR compared to the remaining algorithms in small sets and TGBP in large sets. The channel gain from applying TGBP always less than HCBP, TPK-Means and BK-Means. This is because TGBP prioritizes minimizing the number of beams and balancing the load across beams. As a result, each beam in TGBP can typically serve a larger number of users compared to the other algorithms, leading to decreased channel gain per user. BK-Means, with its iterative clustering approach, can better optimize beam placement and power allocation, resulting in higher SCGNR values. TPK-Means and HCBP also benefit from their approaches, but their high computational time makes them less practical for larger user sets despite their potential for achieving higher SCGNR in smaller sets.

Fig.~\ref{fig:fig3}(c) and Fig.~\ref{fig:fig3}(d) illustrate the running time of each algorithm. TGBP algorithm has a shorter computational time compared to the remaining algorithms, i.e, in the order of milliseconds. Despite a rapid increase in computational time with a rising number of users (from 100 to 1000) leading to the large number of refining steps, TGBP consistently outperforms BK-Means in terms of running time. On the other hand, BK-Means exhibits a gradual and stable increase in computational time due to its iterative K-Means clustering process with a predefined number of iterations. In contrast, HCBP and TPK-Means experience an exponential rise in computational time, making them impractical for large datasets. 

\section{Conclusion}
This paper has illustrated the importance of optimizing beam placement to enhance channel gain for users in the coverage area. Furthermore, load balancing among active beams also helps to save resources while improving communication performance. We have proposed two algorithms to improve computational ability and achieve good feasible solutions in polynomial time compared to the state-of-the-art benchmarks. Numerical experiments illustrated the benefits of two our proposed approaches on channel gain and computation time.  


\bibliographystyle{IEEEtran}
\bibliography{IEEEabrv,refs}

\begin{thebibliography}{10}
\providecommand{\url}[1]{#1}
\csname url@samestyle\endcsname
\providecommand{\newblock}{\relax}
\providecommand{\bibinfo}[2]{#2}
\providecommand{\BIBentrySTDinterwordspacing}{\spaceskip=0pt\relax}
\providecommand{\BIBentryALTinterwordstretchfactor}{4}
\providecommand{\BIBentryALTinterwordspacing}{\spaceskip=\fontdimen2\font plus
\BIBentryALTinterwordstretchfactor\fontdimen3\font minus
  \fontdimen4\font\relax}
\providecommand{\BIBforeignlanguage}[2]{{%
\expandafter\ifx\csname l@#1\endcsname\relax
\typeout{** WARNING: IEEEtran.bst: No hyphenation pattern has been}%
\typeout{** loaded for the language `#1'. Using the pattern for}%
\typeout{** the default language instead.}%
\else
\language=\csname l@#1\endcsname
\fi
#2}}
\providecommand{\BIBdecl}{\relax}
\BIBdecl

\bibitem{9741772}
T.~Darwish, G.~K. Kurt, H.~Yanikomeroglu, M.~Bellemare, and G.~Lamontagne,
  ``{LEO Satellites in 5G and Beyond Networks: A Review From a Standardization
  Perspective},'' \emph{IEEE Access}, vol.~10, pp. 35\,040--35\,060, 2022.

\bibitem{9473743}
S.~Kota and G.~Giambene, ``{6G Integrated Non-Terrestrial Networks: Emerging
  Technologies and Challenges},'' in \emph{2021 IEEE International Conference
  on Communications Workshops}, 2021, pp. 1--6.

\bibitem{10182679}
G.~Saravana~Kumar, K.~K. Ramachandran, S.~Sharma, R.~Ramesh, K.~Qureshi, and
  K.~Ganesh, ``{AI-Assisted Resource Allocation for Improved Business
  Efficiency and Profitability},'' in \emph{2023 3rd International Conference
  on Advance Computing and Innovative Technologies in Engineering (ICACITE)},
  2023, pp. 54--58.

\bibitem{10048927}
D.~U. Kim, S.~B. Park, C.~S. Hong, and E.~N. Huh, ``{Resource Allocation and
  User Association Using Reinforcement Learning via Curriculum in a Wireless
  Network with High User Mobility},'' in \emph{2023 International Conference on
  Information Networking (ICOIN)}, 2023, pp. 382--386.

\bibitem{10366793}
H.-H. Choi, G.~Park, K.~Heo, and K.~Lee, ``{Joint Optimization of Beam
  Placement and Transmit Power for Multibeam LEO Satellite Communication
  Systems},'' \emph{IEEE Internet of Things Journal}, vol.~11, no.~8, pp.
  14\,804--14\,813, 2024.

\bibitem{pachler2021static}
N.~Pachler de~la Osa, M.~Guerster, I.~del Portillo~Barrios, E.~Crawley, and
  B.~Cameron, ``{Static beam placement and frequency plan algorithms for LEO
  constellations},'' \emph{International Journal of Satellite Communications
  and Networking}, vol.~39, no.~1, pp. 65--77, 2021.

\bibitem{10118814}
N.~Torkzaban, A.~Zoulkarni, A.~Gholami, and J.~S. Baras, ``{Capacitated Beam
  Placement for Multi-beam Non-Geostationary Satellite Systems},'' in
  \emph{2023 IEEE Wireless Communications and Networking Conference (WCNC)},
  2023, pp. 1--6.

\bibitem{bui2022}
V.-P. Bui, T.~Van~Chien, E.~Lagunas, J.~Grotz, S.~Chatzinotas, and
  B.~Ottersten, ``{Joint Beam Placement and Load Balancing Optimization for
  Non-Geostationary Satellite Systems},'' in \emph{2022 IEEE International
  Mediterranean Conference on Communications and Networking (MeditCom)}, 2022,
  pp. 316--321.

\bibitem{3gpp2019study}
3GPP, ``{Study on New Radio (NR) to support non-terrestrial networks},''
  \emph{TR 38.811 V15. 2.0}, 2019.

\bibitem{karp2010reducibility}
R.~Karp, \emph{{Reducibility among combinatorial problems}}.\hskip 1em plus
  0.5em minus 0.4em\relax Springer, 2010.

\bibitem{arumugam2016handbook}
S.~Arumugam, A.~Brandst{\"a}dt, T.~Nishizeki \emph{et~al.}, \emph{{Handbook of
  graph theory, combinatorial optimization, and algorithms}}.\hskip 1em plus
  0.5em minus 0.4em\relax CRC Press, 2016, vol.~34.

\end{thebibliography}
\end{document}